\documentclass[reqno,10pt,letterpaper]{amsart}

\usepackage{lipsum}
\usepackage{amsmath}
\usepackage{amssymb}
\usepackage{amsthm}
\usepackage{mathrsfs}
\usepackage{accents}
\usepackage{calc}
\usepackage{arydshln}
\usepackage{upgreek}
\usepackage{slashed}
\usepackage{xifthen}
\usepackage{graphicx}
\usepackage{longtable}
\usepackage[inline]{enumitem}
\usepackage{stmaryrd}

\usepackage{xcolor}
\definecolor{winered}{rgb}{0.6,0,0}
\definecolor{lessblue}{rgb}{0,0,0.7}

\usepackage[pdftex,colorlinks=true,linkcolor=winered,citecolor=lessblue,urlcolor=lessblue,breaklinks=true,bookmarksopen=true]{hyperref}

\makeatletter
\newcommand{\myitem}[2]{\item[\rm(#2)]\def\@currentlabel{#2}\label{#1}}
\makeatother

\usepackage{titletoc}

\makeatletter

\def\@tocline#1#2#3#4#5#6#7{
\begingroup
  \par
    \parindent\z@ \leftskip#3 \relax \advance\leftskip\@tempdima\relax
                  \rightskip\@pnumwidth plus 4em \parfillskip-\@pnumwidth
    \ifcase #1 
       \vskip 0.6em \hskip 0em 
       \or
       \or \hskip 0em 
       \or \hskip 1em 
    \fi%
    #6
    \nobreak\relax{\leavevmode\leaders\hbox{\,.}\hfill}
    \hbox to\@pnumwidth {\@tocpagenum{#7}}
  \par
\endgroup
}

 \def\l@section{\@tocline{0}{0pt}{0pc}{}{}}

\renewcommand{\tocsection}[3]{%
  \indentlabel{\@ifnotempty{#2}{ 
    \ignorespaces\bfseries{#2. #3}}}
  \indentlabel{\@ifempty{#2}{\ignorespaces\bfseries{#3}}{}} 
    \vspace{1.5pt}}

\renewcommand{\tocsubsection}[3]{%
  \indentlabel{\@ifnotempty{#2}{
    \ignorespaces#2. #3}}
  \indentlabel{\@ifempty{#2}{\ignorespaces #3}{}}
    \vspace{1.5pt}}

\renewcommand{\tocsubsubsection}[3]{%
  \indentlabel{\@ifnotempty{#2}{
    \ignorespaces#2. #3}}
  \indentlabel{\@ifempty{#2}{\ignorespaces #3}{}}
    \vspace{1.5pt}}

\makeatother

\makeatletter
\def\@nomenstarted{0}
\newlength{\@nomenoldtabcolsep}

\newcommand{\nomenstart}
  {%
    \def\@nomenstarted{1}%
    \setlength{\@nomenoldtabcolsep}{\tabcolsep}%
    \setlength{\tabcolsep}{3.5pt}%
    \begin{longtable}{p{0.11\textwidth} p{0.86\textwidth}}
  }

\newcommand{\nomenitem}[2]{%
    \ifcase\@nomenstarted%
      \or 
      \or \\ 
    \fi%
    #1\,{\leavevmode\leaders\hbox{\,.}\hfill} & #2%
    \def\@nomenstarted{2}%
  }%
\newcommand{\nomenend}
  {\\%
      \end{longtable}%
      \setlength{\tabcolsep}{\@nomenoldtabcolsep}%
      \def\@nomenstarted{0}%
  }
\makeatother

\makeatletter
\newcommand{\BIG}{\bBigg@{3.5}}
\newcommand{\vast}{\bBigg@{4}}
\newcommand{\Vast}{\bBigg@{5}}
\newcommand{\VAST}[1]{\bBigg@{#1}}
\makeatother

\allowdisplaybreaks

\numberwithin{equation}{section}
\numberwithin{figure}{section}
\newtheorem{thm}{Theorem}[section]

\newtheorem{lemma}[thm]{Lemma}
\newtheorem{cor}[thm]{Corollary}

\newtheorem*{thm*}{Theorem}
\newtheorem*{prop*}{Proposition}
\newtheorem*{cor*}{Corollary}
\newtheorem*{conj*}{Conjecture}

\theoremstyle{definition}
\newtheorem{definition}[thm]{Definition}

\theoremstyle{remark}
\newtheorem{rmk}[thm]{Remark}

\makeatletter
\newcommand{\fakephantomsection}{%
  \Hy@MakeCurrentHref{\@currenvir.\the\Hy@linkcounter}
  \Hy@raisedlink{\hyper@anchorstart{\@currentHref}\hyper@anchorend}%
  \Hy@GlobalStepCount\Hy@linkcounter%
}
\makeatother

\newcommand{\mc}{\mathcal}

\newcommand{\cB}{\mc B}
\newcommand{\cC}{\mc C}

\newcommand{\cE}{\mc E}

\newcommand{\cH}{\mc H}
\newcommand{\cI}{\mc I}

\newcommand{\cL}{\mc L}
\newcommand{\cM}{\mc M}

\newcommand{\cO}{\mc O}

\newcommand{\cU}{\mc U}
\newcommand{\cV}{\mc V}

\newcommand{\cX}{\mc X}

\newcommand{\ms}{\mathscr}

\newcommand{\sC}{\ms C}

\newcommand{\scri}{\ms I}

\newcommand{\N}{\mathbb{N}}
\newcommand{\R}{\mathbb{R}}
\newcommand{\Z}{\mathbb{Z}}

\newcommand{\Sph}{\mathbb{S}}

\newcommand{\sfG}{\mathsf{G}}
\newcommand{\sfH}{\mathsf{H}}

\newcommand{\fm}{\mathfrak{m}}

\newcommand{\slg}{\slashed{g}{}}

\newcommand{\End}{\operatorname{End}}

\newcommand{\Hom}{\operatorname{Hom}}

\newcommand{\mathspan}{\operatorname{span}}

\newcommand{\eps}{\epsilon}

\newcommand{\la}{\langle}

\newcommand{\pa}{\partial}
\newcommand{\dd}{{\mathrm d}}
\newcommand{\ra}{\rangle}

\newcommand{\ubar}[1]{\underaccent{\bar}#1}
\newcommand{\pfstep}[1]{$\bullet$\ \underline{\textit{#1}}}

\newcommand{\Diff}{\mathrm{Diff}}

\newcommand{\CI}{\cC^\infty}

\newcommand{\bhm}{\fm}

\newcommand{\openbigpmatrix}[1]
  {%
    \def\@bigpmatrixsize{#1}%
    \addtolength{\arraycolsep}{-#1}%
    \begin{pmatrix}%
  }
\newcommand{\closebigpmatrix}
  {%
    \end{pmatrix}%
    \addtolength{\arraycolsep}{\@bigpmatrixsize}%
  }

\newlength{\enummargin}

\newcommand{\usref}[1]{{\upshape\ref{#1}}}

\DeclareGraphicsExtensions{.mps}

\makeatletter
\newcommand*{\fwbw}[1]{\expandafter\@fwbw\csname c@#1\endcsname}
\newcommand*{\@fwbw}[1]{\ifcase #1 \or {\rm fw}\or {\rm bw}\fi}
\AddEnumerateCounter{\fwbw}{\@fwbw}
\makeatother

\begin{document}

\title{Horizons of some asymptotically stationary spacetimes}

\date{\today}

\begin{abstract}
  On a class of dynamical spacetimes which are asymptotic as $t\to\infty$ to a stationary spacetime containing a horizon $\cH_0$, we show the existence of a unique null hypersurface $\cH$ which is asymptotic to $\cH_0$. This is a special case of a general unstable manifold theorem for perturbations of flows which translate in time and have a normal sink at an invariant manifold in space. Examples of horizons $\cH_0$ to which our result applies include event horizons of subextremal Kerr and Kerr--Newman black holes as well as event and cosmological horizons of subextremal Kerr--Newman--de~Sitter black holes. In the Kerr(--Newman) case, we show that $\cH$ is equal to the boundary of the black hole region of the dynamical spacetime.
\end{abstract}

\subjclass[2010]{Primary: 37C70. Secondary: 83C57, 53B50}

\author{Peter Hintz}
\address{Department of Mathematics, ETH Z\"urich, R\"amistrasse 101, 8092 Z\"urich, Switzerland}
\email{peter.hintz@math.ethz.ch}

\maketitle

\section{Introduction}
\label{SI}

In this note we give a short proof of (a generalization of) the recent result by Chen--Klainerman \cite{ChenKlainermanHorizon} on the regularity of the future event horizon of perturbations of subextremal Kerr black holes. Underlying our proof is a general unstable manifold theorem based on \cite{HintzPolytrap}. A simple illustrative example, which is a toy model for the time $1$ null-geodesic flow near the event horizon of Kerr in phase space, is as follows. Let $a\in\R$ and consider
\begin{equation}
\label{EqIEx}
  f_0(t,\phi,z) = (t-1,\phi+a-z,z/2),\quad t\in\R,\ \phi\in\Sph^1=\R/2\pi\Z,\ z\in(-1,1).
\end{equation}
This maps $\cH_0:=\R_t\times\Sph^1_\phi\times\{0\}$ into itself. Then we prove that for any perturbation $f(t,\phi,z)=f_0(t,\phi,z)+(0,\cO(t^{-\alpha}),\cO(t^{-\alpha}))$, $\alpha>0$, there exist $t_0\geq 1$ and a unique smooth submanifold $\cH\subset[t_0,\infty)\times\Sph^1\times(-1,1)$ which is $\cO(t^{-\alpha})$-close to $\cH_0$ and so that $f(\cH)\cap\{t\geq t_0\}\subset\cH$.

We set up and prove our general unstable manifold type theorem in~\S\ref{SU}. In~\S\ref{SH}, we put the construction of perturbed horizons into this setup. In essence, we construct a smooth submanifold $\Gamma^{\rm u}$ \emph{of phase space} which asymptotes at late times to the conormal bundle $\Gamma_0$ (projected to the cosphere bundle) of the unperturbed event horizon. The key dynamical input is that the past geodesic flow on subextremal Kerr (and its generalizations), lifted to phase space, is a sink in directions normal to $\Gamma_0$. We then show that $\Gamma^{\rm u}$ projects to the base as (and is the conormal bundle of) a smooth null hypersurface (Theorem~\ref{ThmHPert}). Under mild global assumptions on the perturbed Kerr spacetime, we show that this hypersurface coincides with the event horizon (Theorem~\ref{ThmHCI}).

\section{A general unstable manifold theorem}
\label{SU}

Let $\cX$ be a closed manifold and $\Gamma\subset\cX$ a smooth closed submanifold (in particular compact). Let $\bar f\colon\cX\to\cX$ be a smooth map with $\bar f(\Gamma)=\Gamma$ for which there exists an open neighborhood $\cU\supset\Gamma$ so that $\bar f|_\cU\colon\cU\to\bar f(\cU)$ is a diffeomorphism. Since the differential $D\bar f\colon T_\Gamma\cX\to T_\Gamma\cX$ restricts to a vector bundle map
\begin{equation}
\label{EqIDf}
  \Gamma_p\bar f:=D_p\bar f|_{T_p\Gamma}\colon T_p\Gamma\to T_{\bar f(p)}\Gamma,
\end{equation}
it also induces a map on the normal bundle $N\Gamma=T_\Gamma\cX/T\Gamma$ which we denote
\begin{equation}
\label{EqIDf2}
  N_p\bar f \colon N_p\Gamma \to N_{\bar f(p)}\Gamma.
\end{equation}
Suppose that $\bar f|_\Gamma$ is an almost-isometry and $\Gamma$ is a normal sink for $\bar f$ in the following sense: there exist a fiber metric on $N\Gamma$ and, for all $\eps>0$, a fiber metric on $T\Gamma$ so that
\begin{equation}
\label{EqINormHyp}
  \|N_p\bar f\| < 1;\qquad 1-\eps<m(\Gamma_p\bar f):=\inf_{w\in T_p\Gamma,\,\|w\|=1} \|\Gamma_p\bar f(w)\|\leq\|\Gamma_p\bar f\|<1+\eps.
\end{equation}
Define the map
\begin{equation}
\label{EqIf0}
  f_0 \colon \cM:=\R\times\cX\to\cM,\quad
  f_0(t,x) = (t-1,\bar f(x)),
\end{equation}
which maps $\Gamma_0:=\R\times\Gamma$ to itself. We measure the size of perturbations of $f_0$ using:

\begin{definition}[Weight function]
\label{DefIWeight}
  A weight function is a smooth function $0<\rho\in\CI([0,\infty))$ with
  \begin{equation}
  \label{EqIrho}
    \rho'\leq 0,\quad \lim_{t\to\infty}\rho(t)=0,\quad |\rho^{(k)}|\leq C_k\rho\ \forall\,k\in\N_0.
  \end{equation}
\end{definition}
Examples are $e^{-\alpha t}$ and $(1+t)^{-\alpha}$ for $\alpha>0$. For $m\in\N_0\cup\{\infty\}$, write $\cC_b^m(\cM)$ for the space of $u\in\CI(\cM)$ so that all of its up to $m$ many derivatives along $\pa_t$ and smooth vector fields on $\cX$ are bounded. For a bundle $E\to\cX$, the space $\cC_b^\infty(\cM;E)$ is defined analogously using local trivializations of $E$. Furthermore $\rho\cC_b^m(\cM)=\{u\colon\rho^{-1}u\in\cC_b^m(\cM)\}$.

Let $t_0\geq 0$. Fix a Riemannian metric on $\cX$ and suppose $f\colon\cM\cap\{t\geq t_0\}\to\cM$, is a shift by $-1$ in $t$ and a $\rho\cC_b^\infty$ perturbation of $\bar f$ in space. That is,
\begin{equation}
\label{EqIf}
  f(t,x)=\bigl(t-1,\exp_{\bar f(x)}\tilde V(t,\bar f(x))\bigr),\quad \tilde V(t,\cdot)\in\cV(\cX)=\CI(\cX;T\cX),\ \tilde V\in\rho\CI_b(\cM;T\cX).
\end{equation}

\begin{thm}[Existence and uniqueness of the perturbed unstable manifold]
\label{ThmRMain}
  There exists a smooth submanifold $\Gamma^{\rm u}\subset\cM\cap\{t\geq t_0\}$ with the following properties:
  \begin{enumerate}
  \item $\Gamma^{\rm u}$ is $f$-invariant in the sense that $f(\Gamma^{\rm u})\cap\{t\geq t_0\}=\Gamma^{\rm u}$;
  \item $\Gamma^{\rm u}$ approaches $\Gamma_0$ as $t\to\infty$ in a $\rho\CI_b$ sense, in that $\Gamma^{\rm u}$ is the normal graph over $\Gamma_0$ of a function in $\rho\CI_b(\Gamma_0;N\Gamma)$.
  \end{enumerate}
  Moreover, $\Gamma^{\rm u}$ is the unique $f$-invariant submanifold within the class of manifolds approaching $\Gamma_0$ in a $\rho\cC_b^1$ sense. Finally, given any $\tilde\Gamma^{\rm u}$ which approaches $\Gamma_0$ in a $\tilde\rho\cC_b^1$ sense for some weight function $\tilde\rho$, one obtains $\Gamma^{\rm u}$ as the $\cC_b^0$ limit of iterates $\tilde\Gamma^{\rm u}_0=\tilde\Gamma^{\rm u}\cap\{t\geq t_0\}$, $\tilde\Gamma^{\rm u}_{n+1}=f(\tilde\Gamma^{\rm u}_n)\cap\{t\geq t_0\}$.
\end{thm}
\begin{proof}
  Fix any complementary subbundle $S\to\Gamma$ of $T_\Gamma\cX$ to $T\Gamma$; thus $T_\Gamma\cX\to N\Gamma$ induces
  \begin{equation}
  \label{EqRMainIso}
    S\cong N\Gamma.
  \end{equation}
  We fix a fiber inner product on $S$ by pulling back the fiber inner product on $N\Gamma$.

  If $S$ is invariant under $D\bar f$, i.e.\ $D\bar f\colon S_p\to S_{\bar f(p)}$, then~\eqref{EqINormHyp} implies the $r$-normal hyperbolicity (for every $r$) of $\bar f$ at $\Gamma$, with trivial unstable normal bundle: in the notation of \cite[\S{2.3}]{HintzPolytrap}, one can take $\bar N^u=0$ (trivial bundle over $\Gamma$) and $\bar N^s=S$. Since then all assumptions of \cite[Theorem~2.3]{HintzPolytrap} are satisfied, we obtain the existence and uniqueness of $\Gamma^{\rm u}$ as claimed. The final statement follows immediately from the proof of \cite[Theorem~2.3]{HintzPolytrap} which, following \cite{HirschPughShubInvariantManifolds}, is based on a graph transform argument {\`a} la Hadamard. See in particular \cite[equation~(2.22)]{HintzPolytrap} where the constant denoted $C_\Sigma$ there can be taken arbitrary large but fixed (cf.\ \cite[equation~(2.36)]{HintzPolytrap}).

  To complete the proof, it thus suffices to show that the assumptions~\eqref{EqIDf}--\eqref{EqINormHyp} imply the existence of a smooth $D\bar f$-invariant complementary subbundle $\bar N^s\subset T_\Gamma\cX$. We work in the splitting
  \begin{equation}
  \label{EqISplit}
    T_\Gamma\cX=T\Gamma\oplus S.
  \end{equation}
  Write $\pi_S\colon T_\Gamma\cX\to S$ for the projection, and $S_p\bar f=\pi_S\circ D_p\bar f|_{S_p}\colon S_p\to S_{\bar f(p)}$ for $N_p f$ under the identification~\eqref{EqRMainIso}. The fiber over a point $p\in\Gamma$ of any complementary bundle can be written as
  \begin{equation}
  \label{EqINsFiber}
    \bar N^{\rm s}_p = \{(A_p(w),w)\colon w\in S_p\}
  \end{equation}
  for some $A_p\in\Hom(S_p,T_p\Gamma)$. We thus need to find
  \[
    A\in\Hom(S,T\Gamma)\ \ \text{s.t.}\ \ 
    D_p\bar f\bigl(A_p(w),w\bigr) = \bigl(A_{\bar f(p)}(w'),w'\bigr)\quad \forall\,p\in\Gamma,\ w\in S_p.
  \]
  The desired $A$ is thus an invariant section for the map $\ell\colon A\mapsto A'$ defined by
  \begin{equation}
  \label{EqIellp}
    D_{\bar f^{-1}(p)}\bar f\bigl(A'_{\bar f^{-1}(p)}(w),w\bigr) = \bigl(A_p(w'),w'\bigr).
  \end{equation}
  Define $\Gamma S\bar f$ by writing
  \[
    D_p\bar f = \begin{pmatrix} \Gamma_p\bar f & \Gamma S_p\bar f \\ 0 & S_p\bar f \end{pmatrix}
  \]
  in the splitting~\eqref{EqISplit}. In~\eqref{EqIellp}, we must then have $w'=S_{\bar f^{-1}(p)}\bar f(w)$, and therefore
  \begin{equation}
  \label{EqIEllp2}
    A'_{\bar f^{-1}(p)} = (\Gamma_{\bar f^{-1}(p)}\bar f)^{-1} \circ\bigl( A_p \circ S_{\bar f^{-1}(p)}\bar f - \Gamma S_{\bar f^{-1}(p)}\bar f\bigr).
  \end{equation}
  Note that $\ell\colon A\mapsto A'$ is a fiber bundle map (in fact, affine linear) on the bundle $\Hom(S,T\Gamma)\to\Gamma$ covering $\bar f^{-1}$. We estimate the operator norm $k_p$ of the linearization of $\ell_p\colon A_p\mapsto A'_{\bar f^{-1}(p)}$ (which maps $B\mapsto(\Gamma_{\bar f^{-1}(p)}\bar f)^{-1}\circ B\circ S_{\bar f^{-1}(p)}\bar f$) by
  \[
    k_p \leq \| (\Gamma_{\bar f^{-1}(p)}\bar f)^{-1} \| \|S_{\bar f^{-1}(p)}\bar f \| = m(\Gamma_{\bar f^{-1}(p)}\bar f)^{-1} \| S_{\bar f^{-1}(p)}\bar f \|.
  \]
  For every $r\in\N$ we can choose the fiber metric on $T\Gamma$ so that $\alpha_p:=\|\Gamma_{\bar f^{-1}(p)}\bar f\|<1+\eps$ for any fixed $\eps>0$. By~\eqref{EqINormHyp}, this implies
  \[
    k_p\alpha_p^{r'} \leq \biggl(\,\sup_{p\in\Gamma} \| S_{\bar f^{-1}(p)}\bar f \|\biggr) (1-\eps)^{-1}(1+\eps)^{r'} < 1\quad \forall\,r'=0,\ldots,r
  \]
  if $\eps$ is small enough. \cite[Theorem~3.5]{HirschPughShubInvariantManifolds} now applies (with $r\in\N$ there arbitrary) and produces the desired smooth section $A\in\Hom(S,T\Gamma)$.
\end{proof}

\begin{rmk}[Introductory example]
\label{RmkUEx}
  In the context of example~\eqref{EqIEx}, we take $\cX=\Sph^1_\phi\times\R_z$, $\Gamma=\Sph^1_\phi\times\{0\}$, and $\bar f(\phi,z)=(\phi+a-z,z/2)$; the non-compactness of $\cX$ in $z$ is irrelevant since we are working near $z=0$ anyway. Identifying $N_{(\phi,z)}\Gamma\cong S_{(\phi,z)}:=\mathspan\{\pa_z\}$ and the norms $\|v\pa_\phi\|=|v|$, $\|w\pa_z\|=|w|$, we indeed have~\eqref{EqINormHyp}. The desired $A_p$ is a scalar, and~\eqref{EqIEllp2} reads $A'_{\bar f^{-1}(\phi,0)}=1^{-1}\cdot(A_{(\phi,0)}\cdot\frac12+1)$ with fixed point $A\equiv 2$; and indeed $\bar N^s=\mathspan\{ 2\pa_\phi+\pa_z \}$ is invariant under $D\bar f$.
\end{rmk}

For us, $\bar f$ will arise as the time $1$ flow of a vector field. We proceed to discuss the first condition in~\eqref{EqINormHyp} in this context in some generality. Let $M$ be a smooth manifold, and let $X\subset M$ a smooth compact submanifold. Let $V\in\cV(M)$ be tangent to $X$. Write $\phi_V^t$ for the time $t$ flow of $V$, and $N_x\phi_V^t\colon N_x X\to N_{\phi_V^t(x)}X$ for the map induced by $D_x\phi_V^t\colon T_x M\to T_{\phi_V^t(x)}M$ on the normal bundle. Write $\cI\subset\CI(M)$ for the ideal of functions vanishing on $X$. Define the \emph{linearization} of $V$ at $X$ by
\begin{equation}
\label{EqULin}
  L_V \in \Diff^1(X;N^*X),\quad L_V\colon\dd f\mapsto\dd(V f)\ \ (f\in\cI).
\end{equation}
Since $\CI(X;N^*X)\cong\cI/\cI^2$ via $\dd f\mapsfrom f+\cI^2$, $f\in\cI$, and since $V\colon\cI^j\to\cI^j$, $j=1,2$, this is well-defined. (Equivalently, $L_V\colon\CI(X;N^*X)\ni\xi\mapsto\dd\la\xi,V\ra\in\CI(X;N^*X)$ where $\la\cdot,\cdot\ra$ is the dual pairing.) Furthermore, the Leibniz rule $L_V(g\,\dd f)=g L_V(\dd f)+(V g)\dd f$ holds. Given a local frame $\dd f_1,\ldots,\dd f_k$ of $N^*X$ over $U\subset X$, we thus have (summation convention)
\begin{equation}
\label{EqULinExpl}
  L_V(g^i\,\dd f_i) = (V g^j + \ell^j{}_i g^i)\dd f_j,\quad g_i\in\CI(U),\quad L_V(\dd f_i) =: \ell^j{}_i\dd f_j,\ \ell^i{}_j\in\CI(U).
\end{equation}
Therefore, $L_V$ induces a partial connection on $N^*X$ along integral curves of $V$ in $X$, and therefore a notion of parallel transport $P_V^t\colon N^*_x X\to N^*_{\phi_V^t(x)}X$. (Explicitly, setting $\gamma(s)=\phi_V^s(x)$, we have $P_V^t(t)w=w(t)$ where $w(0)=w=w^i\,\dd f_i$ and $(w^j)'+\ell^j{}_i w^i=0$.)

\begin{lemma}[Parallel transport and normal linearization]
\label{LemmaULin}
  $N_x\phi_V^t=(P_V^t(t)^*)^{-1}\colon N_x X\to N_{\phi_V^t(x)}X$.
\end{lemma}
\begin{proof}
  We begin with a preliminary computation at $x$. For $f\in\cI$, we compute
  \[
    \frac{\dd}{\dd t}(\dd f\circ N_x\phi_V^t)\Big|_{t=0} = \frac{\dd}{\dd t}(\dd f\circ D_x\phi_V^t)\Big|_{t=0} = \frac{\dd}{\dd t}\bigl((\phi_V^t)^*(\dd f)\bigr)\Big|_{t=0} = \cL_V(\dd f)=\dd(V f) = L_V(f).
  \]
  More generally, if $f\in\cI$, $g\in\CI(M)$, then $(V f)|_X=0$ implies
  \[
    \frac{\dd}{\dd t}(g\,\dd f\circ N_x\phi_V^t)\Big|_{t=0} = \dd(i_V(g\,\dd f)) + i_V\dd(g\,\dd f) = g L_V(f) + (V g)\dd f = L_V(g\,\dd f)
  \]
  by the Leibniz rule for $L_V$. Thus $\frac{\dd}{\dd t}(\xi\circ N_x\phi_V^t)|_{t=0}=L_V\xi$ for all $\xi\in\CI(X;N^*X)$. We apply this to $(\phi^{t_0}_V)^*\xi$ where $t_0\in\R$ and observe $L_V\bigl((\phi_V^t)^*\xi\bigr) = \dd\la(\phi_V^t)^*\xi,V\ra = \dd\bigl( (\phi_V^t)^* \la \xi,D_x\phi_V^t(V)\ra \bigr) = (\phi_V^t)^*\dd\la\xi,V\ra = (\phi_V^t)^*L_V\xi$ to get
  \begin{equation}
  \label{EqULinIdent}
    \frac{\dd}{\dd t}(\xi\circ N_x\phi_V^t) = (\phi_V^t)^*L_V\xi.
  \end{equation}

  Let now $n_0\in N_x X$ and $\xi_0\in N^*_x X$. For $\xi:=P_V^t(\xi_0)\in N_{\phi_V^t(x)}X$, we then have
  \[
    \frac{\dd}{\dd t}\la P_V^t(\xi_0),N_x\phi_V^t(n_0)\ra = \frac{\dd}{\dd t} (\xi\circ N_x\phi_V^t)(n_0) = 0
  \]
  by~\eqref{EqULinIdent}. Therefore, $\la\xi_0,n_0\ra=\la P_V^t(\xi_0),N_x\phi_V^t(n_0)\ra$ for all $t$, as was to be shown.
\end{proof}

\begin{cor}[Normal contraction rate]
\label{CorUSink}
  Let $\theta\in\R$. Then the following are equivalent:
  \begin{enumerate}
  \item $\exists$ fiber inner product on $N X$ with $\|N_x\phi_V^t\|\leq e^{-\theta t}$ for all $t\geq 0$, $x\in X$;
  \item $\exists$ fiber inner product on $N^*X$ with $\frac{\dd}{\dd t}\|P_V^t w\|\big|_{t=0}\geq\theta\|w\|$ for all $x\in X$, $w\in N^*_x X$.
  \end{enumerate}
\end{cor}
\begin{proof}
  Take dual inner products on $N X$ and $N^*X$.
\end{proof}

\section{Construction of perturbed horizons}
\label{SH}

We consider Kerr--Newman(--de~Sitter or --anti de~Sitter) spacetimes with mass $\bhm\geq 0$, angular momentum $a$, charge $Q$, and cosmological constant $\Lambda\in\R$. We consider a value $r_+$ of the standard coordinate $r$ at which there is a non-degenerate horizon. Special cases include subextremal Kerr ($\Lambda=0$, and $r_+=\bhm+\sqrt{\bhm^2-a^2}$ is the radius of the event horizon), de~Sitter space ($\bhm=0$, $\Lambda>0$, and $r_+=\sqrt{3/\Lambda}$ is the radius of the cosmological horizon), and subextremal Kerr--(Newman--)de~Sitter ($\Lambda>0$, and $r_+$ is the radius of either the event or the cosmological horizon). In Boyer--Lindquist coordinates, the metric (see e.g.\ \cite[\S{A}]{PodolskyGriffithsKerrNewmanAcc}) is
\begin{align*}
  g_0 &= -\frac{\mu}{b^2\varrho^2}(\dd t-a\sin^2\theta\,\dd\phi)^2 + \varrho^2\Bigl(\frac{\dd r^2}{\mu}+\frac{\dd\theta^2}{\varkappa}\Bigr) + \frac{\varkappa\sin^2\theta}{b^2\varrho^2}(a\,\dd t-(r^2+a^2)\,\dd\phi)^2, \\
    & b=1+\frac{\Lambda a^2}{3},\ \varkappa=1+\frac{\Lambda a^2}{3}\cos^2\theta,\ \varrho^2=r^2+a^2\cos^2\theta, \ \mu=(r^2+a^2)\Bigl(1-\frac{\Lambda r^2}{3}\Bigr) - 2\bhm r+b^2 Q^2.
\end{align*}
We require $b>0$. We consider a value $r_+$ where $\mu(r_+)=0$ and $\mu'(r_+)>0$. These conditions are valid at the event horizon.\footnote{One can equally well consider the case $\mu'(r_+)<0$ which occurs at the cosmological horizon of a de~Sitter type spacetime; we leave it to the reader to track the sign changes in the arguments below.} We pass to Kerr star type coordinates $t_*,r,\theta,\phi_*$ regular near $r=r_+$,
\[
  \dd t=\dd t_*-\frac{b(r^2+a^2)}{\mu}\,\dd r,\quad \dd\phi=\dd\phi_*-\frac{b a}{\mu}\,\dd r.
\]
Writing covectors as $\zeta=\sigma\,\dd t_*+\xi\,\dd r+\eta_\theta\,\dd\theta+\eta_\phi\,\dd\phi_*$, we then compute for $G_0(\zeta)=g_0^{-1}(\zeta,\zeta)$:
\begin{equation}
\label{EqHG0}
  \sfG_0 := \varrho^2 G_0 = \mu\xi^2 + 2 b\bigl( (r^2+a^2)\sigma + a\eta_\phi \bigr)\xi + \sC_0,\quad \sC_0 := \varkappa\eta_\theta^2 + \frac{b^2}{\varkappa\sin^2\theta}(\eta_\phi+a\sin^2\theta\,\sigma)^2.
\end{equation}
Since $\varrho^2 g_0^{-1}(\dd r,\dd r)=\mu=0$ at $r=r_+$, the hypersurface
\begin{equation}
\label{EqH0}
  \cH_0 := \R_{t_*} \times \Sph^2_{\theta,\phi_*} \times \{r_+\} \subset M = \R_{t_*}\times\Sph^2_{\theta,\phi_*} \times (r_+-2\delta,r_++2\delta)
\end{equation}
is a null hypersurface (the \emph{event horizon} for $g_0$), and $G_0$ vanishes on $N^*\cH_0$, in particular on
\[
  N_-^*\cH_0=\{(t_*,r,\theta,\phi_*,\sigma,\xi,\eta_\theta,\eta_\phi)\colon r=r_+,\ \sigma=\eta_\theta=\eta_\phi=0,\ \xi<0\}.
\]
We consider here only the half $\xi<0$ of $N^*\cH_0\setminus o$ since $H_{\sfG_0}t_*=\pa_\sigma\sfG_0=2 b(r_+^2+a^2)\xi<0$ there. It will be convenient to work on the spherical cotangent bundle $S^*M=(T^*M\setminus o)/\R_+$ near $S N_-^*\cH_0$ (the image of $N_-^*\cH_0$ in $S^*M$) by introducing the fiber coordinates
\[
  \hat\sigma=\frac{\sigma}{\xi},\ \hat\eta_\theta=\frac{\eta_\theta}{\xi},\ \hat\eta_\phi=\frac{\eta_\phi}{\xi}.
\]
For now, we keep track also of the inverse fiber-radial coordinate $\rho_\infty=-\frac{1}{\xi}>0$. Then
\begin{equation}
\label{EqHG0Vf}
\begin{split}
  \rho_\infty H_{\sfG_0} &= -2 b\bigl( (r^2+a^2)\pa_{t_*} + a\pa_{\phi_*}\bigr) - 2\bigl(\mu+b((r^2+a^2)\hat\sigma+a\hat\eta_\phi)\bigr)\pa_r \\
    &\qquad - (\mu'+4 b r\hat\sigma)(\rho_\infty\pa_{\rho_\infty} + \hat\sigma\pa_{\hat\sigma} + \hat\eta_\theta\pa_{\hat\eta_\theta} + \hat\eta_\phi\pa_{\hat\eta_\phi} ) + \rho_\infty H_{\sC_0}; \\
  \sfH_{\sC_0} := \rho_\infty H_{\sC_0} &= -2\varkappa\hat\eta_\theta\pa_\theta -\frac{2 b^2}{\varkappa\sin^2\theta}(\hat\eta_\phi+a\sin^2\theta\,\hat\sigma)(\pa_{\phi_*}+a\sin^2\theta\,\pa_{t_*}) \\
    &\qquad + \Bigl(\varkappa'\hat\eta_\theta^2 + \pa_\theta\Bigl(\frac{b^2}{\varkappa\sin^2\theta}(\hat\eta_\phi+a\sin^2\theta\,\hat\sigma)^2\Bigr)\Bigr)\pa_{\hat\eta_\theta}
\end{split}
\end{equation}
By homogeneity, these induce vector fields on $S^*M$: this removes the term involving $\rho_\infty\pa_{\rho_\infty}$. We now restrict to the characteristic set $\sfG_0^{-1}(0)$: note that $\pa_r\sfG_0=\mu'\xi^2>0$ at $N^*\cH_0$, so $r$ can be written as a function of $\sigma,\xi,\eta_\theta,\eta_\phi,\theta,\phi_*$ in a conic neighborhood, and thus (by homogeneity) as a function $r=r_0(\hat\sigma,\hat\eta_\theta,\hat\eta_\phi,\theta,\phi_*)$. Restricting to $\{\rho_\infty^2\sfG_0=0\}$ thus allows us to further drop the $\pa_r$ component. Finally normalizing the coefficient of $\pa_{t_*}$ to be $-1$, we thus define
\begin{equation}
\label{EqHVf}
\begin{split}
  V_0 &:= \frac{1}{2 b(r^2+a^2)}(\rho_\infty H_{\sfG_0})|_{(\rho_\infty^2\sfG_0)^{-1}(0)} = -\pa_{t_*} + \bar V, \\
  \bar V &= -\frac{\mu'(r)+4 b r\hat\sigma}{2 b(r^2+a^2)}(\hat\sigma\pa_{\hat\sigma} + \hat\eta_\theta\pa_{\hat\eta_\theta} + \hat\eta_\phi\pa_{\hat\eta_\phi} ) - \frac{a}{r^2+a^2}\pa_{\phi_*} + \frac{1}{2 b(r^2+a^2)}\sfH_{\sC_0},
\end{split}
\end{equation}
where $\bar V$ is tangent to
\begin{equation}
\label{EqHGammaX}
  \Gamma := \{ \hat\sigma=\hat\eta_\theta=\hat\eta_\phi=0 \} \subset \cX := \R_{\hat\sigma} \times (T^*\Sph^2)_{\theta,\phi_*,\hat\eta_\theta,\hat\eta_\phi}.
\end{equation}
Thus, $\Gamma$ is a cross section of $S N_-^*\cH_0$ projected off the $r$-coordinate; and $\bar V\in\cV(\cX)$. Let $\cU\subset\cX$ denote a precompact open neighborhood of $\Gamma$ so that the time $1$ flow $e^{\bar V}$ of $\bar V$ remains in $\cX$. Define\footnote{If we smoothly cut off $\bar V$ to $0$ outside of an open neighborhood of $\Gamma$ so that the resulting vector field $\bar V'$ vanishes outside of $\cU$, then we can also replace $\cX$ by a compact manifold, and $e^{\bar V'}\colon\cX\to\cX$ diffeomorphically. This puts us into the setting of~\S\ref{SU}. Since all considerations for now are local near $\Gamma$, we permit ourselves the minor imprecision of working with~\eqref{EqHGammaX}.}
\[
  \bar f := e^{\bar V} \colon \cU\to\cX.
\]

\begin{lemma}[Normal sink nature of $\Gamma$]
\label{LemmaH1}
  Fix the standard metric on $\Sph^2$ on $\Gamma=\Sph^2_{\theta,\phi_*}$. Then $\bar f|_\Gamma$ is an isometry. There exists a fiber metric on $N\Gamma$ so that $\|N_p\bar f\|<1$, $p\in\Gamma$, and thus~\eqref{EqINormHyp} holds.
\end{lemma}
\begin{proof}
  The first part follows from $\bar f|_\Gamma\colon(\theta,\phi_*)\mapsto(\theta,\phi_*-\frac{a}{r_+^2+a^2})$ being a rotation. For the second part, we write $\kappa:=\frac{\mu'(r_+)}{2 b(r_+^2+a^2)}>0$ and compute the linearization $L_{\bar V}$ of $\bar V$ at $\Gamma$ (see~\eqref{EqULin}) to map $\xi\mapsto-\kappa\xi$ for $\xi=\dd\hat\sigma$, $\dd\hat\eta_\theta$, $\dd\hat\eta_\phi$. We thus read off from $L_{\bar V}$ the bundle endomorphism
  \[
    \End(N^*\Gamma) = \End\bigl( \R\dd\hat\sigma \oplus T^*_o(T^*\Sph^2) \bigr) \ni \ell = -\kappa.
  \]
  This is the map encoded by $\ell^i{}_j$ in~\eqref{EqULinExpl} upon using as a local frame $\dd\hat\sigma$ and differentials of fiber-linear functions on $T^*\Sph^2$ (for a cover of $\Sph^2$ by coordinate charts).

  Fix the fiber inner product $\la\cdot,\cdot\ra:=I\oplus\slg$ where $\slg$ is the standard inner product on $T^*_o(T^*\Sph^2)$ coming from $T^*\Sph^2$; then $\la w,\ell w\ra=-\kappa\la w,w\ra$. Let $\|u\|:=\la u,u\ra^{1/2}$. Given $w\in N^*_x\Gamma$, we then have for $w':=-\ell w$ the inequality $\frac{\dd}{\dd t}\|w+t w'\|\big|_{t=0}\geq\kappa\|w\|$ (in fact, equality holds). Corollary~\ref{CorUSink} gives $\|N_p\bar f\|=\|N_p\phi_{\bar V}^1\|\leq e^{-\kappa}<1$, as desired. (This is a somewhat refined version of arguments appearing in Vasy~\cite[\S{6.3}]{VasyMicroKerrdS} and, in greater generality, Gannot \cite{GannotHorizons}.)
\end{proof}

\begin{thm}[Construction of the perturbed horizon]
\label{ThmHPert}
  Let $\rho$ be a weight function. Let $g$ be a Lorentzian metric on $M=\R_{t_*}\times X$, $X=\Sph^2\times(r_+-2\delta,r_++2\delta)$, so that $g-g_0\in\rho\CI_b$ on $\{t_*\geq 0\}$, i.e.\ the metric coefficients in stationary frames consisting of $\pa_{t_*}$ and local frames of $T X$ are of class $\rho\CI_b$. Then there exists $t_0$ so that there exists a unique smooth null hypersurface $\cH\subset[t_0,\infty)\times X$ which approaches $\cH_0$ in a $\rho\CI_b$ sense. One can take $t_0=0$ when $\|g-g_0\|_{\rho\cC_b^2}$ is sufficiently small.
\end{thm}
\begin{proof}
  \pfstep{Construction of $\cH$.} For a subset $A\subset S^*M$, denote by $C(A)\subset T^*M\setminus o$ its conic extension. Let $G(\zeta)=g^{-1}(\zeta,\zeta)$, $\sfG=\varrho^2 G$. For a $t_*$-independent neighborhood $\cU_0$ of $\Gamma_0=\R_{t_*}\times\Gamma$, we have $\rho_\infty H_\sfG t_*<0$ on $\cU_0$ and $\pa_r\sfG>0$ on $C(\cU_0)$. The latter allows us to write $\{\rho_\infty^2\sfG=0\}\cap\cU_0$ as the graph of a smooth function $r\colon\cX\to(r_+-2\delta,r_++2\delta)$ approaching the analogous function $r_+$ for $\rho_\infty^2\sfG_0$ in a $\rho\CI_b$ sense. (See \cite[Lemma~4.6]{HintzPolytrap} for a similar argument.) Write $F\colon\cX\to\cX\times(r_+-2\delta,r_++2\delta)$, $x\mapsto(x,r(x))$, and define $F_0$ analogously using $r_0$ in place of $r$. Then
  \[
    V := F^*\Bigl(\frac{1}{-\rho_\infty H_\sfG t_*}(\rho_\infty H_\sfG)|_{S^*M}\Bigr) \in \cV(\R_{t_*}\times\cX)
  \]
  is a $\rho\CI_b$-perturbation of $V_0$ in~\eqref{EqHVf} and satisfies $V t_*=-1$. By Lemma~\ref{LemmaH1}, we can apply Theorem~\ref{ThmRMain} to its time $1$ flow $f=\phi_V^1$ to obtain a smooth $f$-invariant submanifold
  \[
    \Gamma^{\rm u} \subset [t_0,\infty)\times\cX
  \]
  approaching $\Gamma_0$ in a $\rho\CI_b$ sense. Applying the same result, for $N t_*$ in place of $t_*$, to the time $1/N$ flow of $V$ shows that $\Gamma^{\rm u}$ is also $\phi^{1/N}_V$-invariant; taking $N\to\infty$ implies that $V$ is tangent to $\Gamma^{\rm u}$.

  We now map $\Gamma^{\rm u}$ back into $\{\rho_\infty^2\sfG=0\}$ via $F$. Thus $C(F(\Gamma^{\rm u}))\subset\sfG^{-1}(0)$, and $H_\sfG$ is tangent to $C(F(\Gamma^{\rm u}))$. Increasing $t_0$ if necessary, the base projection $\pi\colon C(F(\Gamma^{\rm u}))\to[t_0,\infty)\times X$ has rank $3$ since this is true for $\pi\colon C(F_0(\Gamma_0))=N^*\cH_0\setminus o\to M$; its image is thus a smooth hypersurface
  \[
    \cH = \pi(C(F(\Gamma^{\rm u}))) \subset \{t_*\geq t_0\} \cap M.
  \]
  By construction, for $z_0\in\cH$ and $0\neq\zeta_0\in T^*_{z_0}M$ with $(z_0,\zeta_0)\in C(F(\Gamma^{\rm u}))$, the $H_\sfG$-integral curve starting at $(z_0,\zeta_0)$ remains in $C(F(\Gamma^{\rm u}))$ as long as it remains in $\{t_*\geq t_0\}$. Therefore, its projection along $\pi$, which is a null-geodesic, remains in $\cH$. We draw two conclusions: first, $T_{z_0}\cH$ is null or Lorentzian, as it contains a null vector $\zeta_0^\sharp$; second,
  \begin{equation}
  \label{EqHPertRule}
    \text{$\cH$ is ruled by null-geodesics.}
  \end{equation}

  \pfstep{$\cH$ is null: reduction.} The main task is to show that $\cH$ (i.e.\ each $T_{z_0}\cH$) is null. It suffices to prove this for $\cH\cap\{t_*\geq t_1\}$ for any $t_1\geq t_0$. We recall the classical argument. First, let $L\in\cV(\cH)$ be a positive multiple of $\zeta_0^\sharp$ at each $z_0\in\cH$ with $(z_0,\zeta_0)\in C(F(\Gamma^{\rm u}))$. Then the integral curves of $L$ are reparameterized past null-geodesics, so $\nabla_L L=c L$, $c\in\CI(\cH)$. For $X\in\CI(\cH\cap\{t_*=t_1\};T\cH)$, define $\tilde X\in\cV(\cH\cap\{t_0\leq t_*\leq t_1\})$, $\tilde X|_{t_*=t_1}=X$, via $\cL_L\tilde X=0$. Then $L g(L,\tilde X)=g(\nabla_L L,\tilde X)+g(L,[L,\tilde X])+\frac12\tilde X g(L,L)=c g(L,\tilde X)$. Since $g(L,\tilde X)=g(L,X)=0$ at $t_*=t_1$, we obtain $g(L,\tilde X)\equiv 0$. Since for $z_0\in\cH\cap\{t_0\leq t_*\leq t_1\}$, the space of $\tilde X|_{z_0}\in T_{z_0}\cH$ for such $X$ is all of $T_{z_0}\cH$, this shows that $T_{z_0}\cH=L^\perp$ is null indeed.

  \pfstep{Decomposition of $M$ in $t_*\geq t_1$.} Fix now $t_1\geq t_0$ so that $\dd r$ is timelike at $r=r_+-\delta$ for $t_*\geq t_1$, and so that at each point $z_0\in M\cap\{t_*\geq t_1\}$ with $r(z_0)=r_++\delta$ there exist future null vectors $W_+,W_0\in T_{z_0}M$, with smooth dependence on $z_0$, so that $\dd r(W_+)>0$, $\dd r(W_0)=0$. (Since such $W_+,W_0$ exist for the metric $g_0$, a time $t_1$ with this property exists.) Define the sets
  \begin{equation}
  \label{EqHPertBE}
  \begin{split}
    B &:= \bigl\{ z\in M\cap\{t_*\geq t_1\} \colon \forall\,\text{maximal future null-geodesics}\ \gamma\colon[0,\bar s)\to M,\ \gamma(0)=z, \\
      &\hspace{15em} \text{we have}\ \ubar r:=\liminf_{s\to\bar s} r(\gamma(s))<r_+ \bigr\}, \\
    E &:= \bigl\{ z\in M\cap\{t_*\geq t_1\} \colon \exists\,\text{future causal curve}\ \gamma\colon[0,\bar s)\to M,\ \gamma(0)=z, \\
      &\hspace{15em} \text{so that}\ \lim_{s\to\bar s} t_*(\gamma(s))=\infty,\ \bar r:=\limsup_{s\to\bar s} r(\gamma(s))>r_+ \bigr\}.
  \end{split}
  \end{equation}
  We shall prove that
  \begin{equation}
  \label{EqHPertDecomp}
    M\cap\{t_*\geq t_1\}=B\sqcup(\cH\cap\{t_*\geq t_1\})\sqcup E.
  \end{equation}
  Given $z\in B$ and a maximal future null-geodesic $\gamma\colon[0,\bar s)\to M$ starting at $z$, suppose that $t_*\circ\gamma$ is bounded; then $\gamma$ must exit $M$, so $\ubar r=r_+-2\delta$. If $t_*\circ\gamma$ is unbounded, we use that $\dd r$ is timelike for $r\leq r'$ for any fixed $r'<r_+$ when $t_*$ is sufficiently large: this forces $r\circ\gamma$ to be monotonically decreasing for sufficiently large $s$, and $\gamma$ exits $M$ through $r=r_+-2\delta$ again. Thus, $B$ is open and disjoint from $E$. Moreover, $B\supset\{r=r_+-\delta,\ t_*\geq t_1\}$ is non-empty.

  Similarly, when $r'>r_+$ is fixed, $r(\gamma(s))\geq r'$, and $t_*(\gamma(s))$ is large enough depending on $r'$, then $\gamma$ can be continued past $\gamma(s)$ as a future timelike curve which reaches $\{r=r_++\delta\}$ since the same is true on $(M,g_0)$. This implies that $E$ is open. Moreover, $E\supset\{r=r_++\delta,\ t_*\geq t_1\}$ is non-empty.

  Let now $z\in(M\cap\{t_*\geq t_1\})\setminus(B\cup E)$. Since $z\notin B$, there exists a maximal future null-geodesic $\gamma$ with $\gamma(0)=z$ and $\ubar r\geq r_+$. Since $z\notin E$, $t_*\circ\gamma$ is bounded or $\bar r\leq r_+$. If $\bar r>r_+$ and $t_*\circ\gamma$ is bounded, then $r\circ\gamma$ must leave every compact subset of $(r_+-2\delta,r_++2\delta)$, which in view of $\ubar r\geq r_+$ forces $r(\gamma(s))\geq r_++\delta$ for some $s\in[0,\bar s)$. Continuing $\gamma$ past $\gamma(s)$ as an integral curve of $W_0$ produces a future causal curve starting at $z$ along which $t_*$ is unbounded and $\bar r=r_++\delta$, in contradiction to $z\notin E$. Therefore, $\bar r\leq r_+$ and thus
  \[
    \lim_{s\to\bar s} r(\gamma(s))=r_+,\quad \lim_{s\to\bar s} t_*(\gamma(s))=\infty.
  \]
  Lift $\gamma$ to phase space as an integral curve of $H_G$ and project it to $S^*M$. Call the resulting curve $\alpha\colon[0,\bar s)\to S^*M$; we reparameterize is so that $(t_*\circ\alpha)'=1$. The projection $[\alpha]$ of $\alpha$ to the quotient $S^*M/\R$ of $S^*M$ by $t_*$-translations remains in a compact set. Denote by $\varpi\in S^*M/\R$ a point in the $\omega$-limit set of $[\alpha]$; it satisfies $r(\varpi)=r_+$. We claim that
  \begin{equation}
  \label{EqHCIvarpi}
    \varpi\in S N_+^*\cH_0/\R
  \end{equation}
  where $N_+^*\cH_0$ is the \emph{future} component, with $\xi>0$, of the spherical conormal bundle of $\cH_0$. Denote by $L^+\subset S^*M\cap\{r\leq r_+\}$ the projection to $S^*M$ of the future component of $\sfG_0^{-1}(0)\cap(T^*M\setminus o)$ in $r\leq r_+$. To show~\eqref{EqHCIvarpi}, note that $\rho_\infty H_{\sfG_0}r<0$ has a strictly negative upper bound in a $t_*$-translation invariant neighborhood $\cU_0\subset S^*M$ of $\{r\leq r_+\}\cap(L^+\setminus\cV)$, where $\cV$ is any fixed $t_*$-translation invariant neighborhood of $S N_+^*\cH_0$. Indeed, in $r<r_+$ we have $H_{\sfG_0}r<0$ along future null-bicharacteristics since $\dd r$ is timelike there due to $\mu(r)<0$, cf.\ \eqref{EqHG0}. By continuity, $H_{\sfG_0}r\leq 0$ in $r=r_+$ in the future characteristic set; but $H_{\sfG_0}r=0$ and $r=r_+$ imply by~\eqref{EqHG0Vf} that $(r^2+a^2)\sigma+a\eta_\phi=0$, hence from $\sfG_0=0$ also $\sC_0=0$ by~\eqref{EqHG0}. This implies $\sigma=\eta_\theta=\eta_\phi=0$, so one is at $N^*\cH_0$. --- We then also have a strict negative upper bound for $\rho_\infty H_\sfG r$ on $\cU_0\cap\sfG^{-1}(0)\cap\{t_*\geq t_2\}$ for sufficiently large $t_2\geq t_1$. If~\eqref{EqHCIvarpi} were false, we could choose $\cV$ to have closure disjoint from $\varpi$; we would then obtain a contradiction to $\ubar r\geq r_+$.

  Having established~\eqref{EqHCIvarpi} for all points in the $\omega$-limit set of $[\alpha]$, we infer that $\alpha$ in fact converges to $S N_+^*\cH_0$, and thus $-\alpha$ (fiberwise minus sign) converges to $S N_-^*\cH_0$ as $t_*\to\infty$; and its image is an integral curve of $\rho_\infty H_\sfG$. Upon shifting the argument of $-\alpha$ by a constant, we thus have $F^{-1}((-\alpha)(s))=(s,\hat\sigma(s),\omega(s),\hat\eta(s))\in\R_{t_*}\times\R_{\hat\sigma}\times(T^*\Sph^2)_{\omega,\hat\eta}$ where $|\hat\sigma(s)|+|\hat\eta(s)|\to 0$ as $s\to\infty$. Using Lemma~\ref{LemmaHDec} below, there exists a weight function $\tilde\rho$ so that $|\hat\sigma(s)|+|\hat\eta(s)|\leq\tilde\rho(s)$. In the application of Theorem~\ref{ThmRMain} in the proof of Theorem~\ref{ThmHPert}, we may then start with $\tilde\Gamma^{\rm u}$ which approaches $\Gamma_0$ in a $\tilde\rho\cC_b^1$ sense and contains the image $A$ of $F^{-1}\circ\alpha$ intersected with $\{t_*\geq t_1\}$. Since $\alpha$ is an integral curve of $H_\sfG$, every iterate $\tilde\Gamma^{\rm u}_n$ in the notation of Theorem~\ref{ThmRMain} contains $A$, and thus so does the limit $\Gamma^{\rm u}$. Therefore, $\gamma([0,\bar s))\subset\cH$, and hence $z=\gamma(0)\in\cH$.

  Let now $\nu\in\CI(M\cap\{t_*\geq t_1\})$ denote a defining function of $\cH$, i.e.\ $\cH=\{\nu=0\}$, $\dd\nu\neq 0$ on $\cH$, and $\pm\nu>0$ at $\{r=r_+\pm\delta\}$. Define the (open) sets $B_\pm:=B\cap\{\pm\nu>0\}$, $E_\pm:=E\cap\{\pm\nu>0\}$. We have $\{\nu<0\}=B_-\sqcup E_-$; since $B_-\neq\emptyset$, we get $E_-=\emptyset$ and $B\supset B_-=\{\nu<0\}$. This implies $E\cap\cH=\emptyset$ since the neighborhood of a point in $E\cap\cH$ would contain a point lying both in $\{\nu<0\}\subset B$ and in $E$. By~\eqref{EqHPertRule}, we have $B\cap\cH=\emptyset$. Finally, $\{\nu>0\}=B_+\sqcup E_+$ with $E_+\neq\emptyset$, and thus $B_+=\emptyset$, $E\supset E_+=\{\nu>0\}$. This gives $B=B_-$, $E=E_+$, and proves~\eqref{EqHPertDecomp}.

  \pfstep{$\cH$ is null: conclusion.} Let $z_0\in\cH\cap\{t_*\geq t_1\}$. If $T_{z_0}\cH\subset T_{z_0}M$ were not a lightlike hyperplane, it would contain a future timelike vector. Since the set of future timelike vectors is open, there would exist a future timelike vector $W\in T_{z_0}M$ pointing into $E$ (i.e.\ $\dd\nu(W)>0$), and hence a causal curve $\gamma$ with $\gamma(0)=z_0$ and $\gamma'(0)=W$ would enter $E$. This would imply $z_0\in E$, contradicting~\eqref{EqHPertDecomp}.

  \pfstep{Uniqueness of $\cH$.} Let $\cH$ be a null hypersurface which approaches $\cH_0$ in a $\rho\cC_b^\infty$ sense. Note that $\xi<0$ on one component of $\{t_*\geq t_0\}\cap(N^*\cH\setminus o)$ when $t_0$ is sufficiently large. Thus $S N_-^*\cH$ approaches $S N_-^*\cH_0$ in a $\rho\CI_b$ sense, and hence $F^{-1}(S N_-^*\cH)$ approaches $\Gamma_0$ in a $\rho\CI_b$ sense. The uniqueness part of Theorem~\ref{ThmRMain} now yields the claim. (This argument also gives $C(F(\Gamma^{\rm u}))=N^*\cH\setminus o$.)

  Finally, the fact that one can take $t_0=0$ when $\|g-g_0\|_{\rho\cC_b^2}$ is sufficiently small, and thus the vector fields $\rho_\infty H_{G_0}|_{S^*M}$ and $\rho_\infty H_G|_{S^*M}$ are close in $\rho\cC_b^1$, follows from an inspection of the proof of \cite[Theorem~2.3]{HintzPolytrap}.
\end{proof}

\begin{lemma}[Weight function]
\label{LemmaHDec}
  Let $d\colon[0,\infty)\to[0,\infty)$ with $\lim_{s\to\infty}d(s)=0$. Then there exists a weight function $\rho$ with $d(s)\leq\rho(s)$ for all $s$.
\end{lemma}
\begin{proof}
  Replace $d$ by $s\mapsto\sup_{s'\geq s}d(s')$; then $d$ is monotonically decreasing. Define the sequence
  \[
    \bar d(n) := \max\{ 2^{-k}d(n-k) \colon k=0,1,\ldots,n \},\quad n\in\N_0.
  \]
  Then $\bar d(n)\geq d(n)$ for all $n$, and $\bar d(n)$ does not fall off too quickly in that
  \[
    \bar d(n+1) = \max\bigl( d(n+1), \max \{ 2^{-k-1}d(n-k) \colon k=0,\ldots,n \} \bigr) \geq \frac12\bar d(n).
  \]
  Fix $\phi\in\CI([0,1])$ so that $\phi=1$ on $[0,\frac13]$, $\phi'<0$ on $(\frac13,\frac23)$, and $\phi=0$ on $[\frac23,1]$. Set
  \[
    \rho(n+\eta) = 2\bigl(\phi(\eta)\bar d(n) + (1-\phi(\eta))\bar d(n+1)\bigr),\quad n\in\N_0,\ \eta\in[0,1).
  \]
  This is equal to $\bar d(n)$ on $[\max(n-\frac13,0),n+\frac13]$ and thus smooth. Moreover, $\rho(n+\eta)\geq 2\phi(\eta)\bar d(n)+(1-\phi(\eta))\bar d(n)\geq\bar d(n)\geq d(n)\geq d(n+\eta)$, and
  \[
    |\rho^{(k)}(n+\eta)| = 2|\phi^{(k)}(\eta)| \bigl(\bar d(n)-\bar d(n+1)\bigr) \leq 2|\phi^{(k)}(\eta)|\bar d(n) \leq 2\|\phi\|_{\cC^k}\rho(n+\eta).\qedhere
  \]
\end{proof}

While $\cH$ in Theorem~\ref{ThmHPert} is constructed locally, the definition of the event horizon of an asymptotically flat spacetime with future affine complete null infinity $\scri^+$---namely, the boundary of the causal past $J^-(\scri^+)$ of $\scri^+$---is global. We show that when $(M,g_0)$ is subextremal Kerr--Newman, then $\cH$ is the event horizon of $(M,g)$ under mild assumptions on the spacetime. (These assumptions are satisfied by the spacetimes arising in nonlinear stability proofs \cite{DafermosHolzegelRodnianskiTaylorSchwarzschild,KlainermanSzeftelKerr,GiorgiKlainermanSzeftelStability,ShenGCMKerr}.)

\begin{thm}[Smoothness of the event horizon of asymptotically subextremal Kerr--Newman spacetimes]
\label{ThmHCI}
  Let $g_0$ denote a subextremal Kerr--Newman spacetime defined on $M=\R_{t_*}\times X$, $X=[r_+-2\delta,\infty)\times\Sph^2$, with the level sets of $t_*$ spacelike and transversal to the future event horizon $\cH_0=\{r=r_+\}$. Let $g$ be a Lorentzian metric on $[0,\infty)\times X$ so that $g-g_0\in\rho\CI_b$ (in the frame on $\R_{t_*}\times\R^3\supset\R_{t_*}\times X$ given by coordinate vector fields). Assume moreover that there exists a radius $R_0$ so that every point in $\{t_*\geq 0,\ r\geq R_0\}$ can be connected to future null infinity $\scri^+$ by a causal curve. Define $\cH\subset\{t_*\geq t_0\}$ using Theorem~\usref{ThmHPert}, and denote by $\cB$, resp.\ $\cE$ the component of $(M\cap\{t_*\geq t_0\})\setminus\cH$ on which $r$ is bounded, resp.\ unbounded. Then:
 \begin{enumerate}
  \item every future causal curve starting in $\cB$ remains there;
  \item starting at every point in $\cE$, there exists a future causal curve tending to a point at $\scri^+$;
  \item the only future causal curves remaining in $\cH$ are the null-generators of $\cH$.
  \end{enumerate}
  In particular, $\cH$ is the event horizon of $(M,g)$.
\end{thm}

An analogous result holds (by the same proof), \emph{mutatis mutandis}, on asymptotically subextremal Kerr--Newman--de~Sitter and --anti de~Sitter spacetimes, with the event horizon defined to be the boundary of the causal past of the conformal boundary.

\begin{rmk}[The radius $R_0$]
\label{RmkHCIR0}
  The existence of $R_0$ is a condition only on the far field behavior of $g$. If $g_0$ has mass parameter $\bhm\in\R$, then this condition is satisfied, for example, when there exists $\delta>0$ so that, in Boyer--Lindquist coordinates and using $F(r):=1-\frac{2\bhm}{r}$, $h:=g-g_0$,
  \begin{equation}
  \label{EqHCIR0}
    h(\pa_t+F\pa_r,\pa_t+F\pa_r)=\cO(r^{-1-\delta}),\ h(\pa_t,\pa_r)=o(1),\ h(\pa_r,\pa_r)=o(1),\quad r\to\infty.
  \end{equation}
  Taking $g_0$ to be the Schwarzschild metric (Kerr--Newman being an $\cO(r^{-2})$ perturbation thereof) and using the tortoise coordinate $r_*$ with $\dd r_*=\frac{\dd r}{F}$, the radial curve $\gamma(s)=(t(s),r_*(s)):=(t+s,R_0+s-a(s))$ with $\dot\gamma(s)=\pa_t+(1-\dot a)\pa_{r_*}=\pa_t+(1-\dot a)F\pa_r$ is timelike if and only if $F^{-1}g_0(\dot\gamma(s),\dot\gamma(s))=-1+(1-\dot a)^2=-2\dot a+\dot a^2<0$, i.e.\ $\dot a\in(0,2)$; and it limits to a point at null infinity if $\lim_{s\to\infty} a(s)<\infty$. A possible choice is thus $\dot a=s^{-1-\delta/2}$; and then $g_0(\dot\gamma(s),\dot\gamma(s))<-s^{-1-\delta/2}$ for large enough $s$. In view of~\eqref{EqHCIR0}, we then have $g(\dot\gamma(s),\dot\gamma(s))<-\frac12 s^{-1-\delta/2}$ for large $s$, as desired.
\end{rmk}

\begin{proof}[Proof of Theorem~\usref{ThmHCI}]
  We define $B,E$ by~\eqref{EqHPertBE} with $t_0$ in place of $t_1$. Given $r'>r_+$, note that on $(M,g_0)$ we can connect any point $z'\in M$ with $r(z')=r'$ to a point with $r\geq R_0$ using a future timelike curve with compact interval of definition; the same thus remains true on $(M,g)$ when, in addition, $t_*(z')$ is sufficiently large. Applying this for $z\in E$ and a future causal curve $\gamma\colon[0,\bar s)\to M$, with $\gamma(0)=z$ and $\bar r=\limsup_{s\to\bar s}r(\gamma(s))>r_+$, to $z'=\gamma(s)$ where $t_*(\gamma(s))$ is large and $r(\gamma(s))\geq r':=\frac{r_++\bar r}{2}$, we conclude that $z$ can be connected to $\scri^+$ by a causal curve. Therefore, $E\subset J^-(\scri^+)\cap\{t_*\geq t_0\}$, whereas $(B\cup\cH)\cap J^-(\scri^+)=\emptyset$. This implies that $\cH$ is the boundary of $J^-(\scri^+)$ in $\{t_*\geq t_0\}$, and $\cB=B$, $\cE=E$. The proof is complete.
\end{proof}

\bibliographystyle{alphaurl}


\end{document}